\documentclass{article}[10pt]
\usepackage{amsfonts}
\usepackage{epsfig}
\usepackage{amssymb,amsmath,bm}
\usepackage{subfig}
\usepackage{amsthm}
\usepackage{mathrsfs}
\usepackage[auth-sc,affil-it]{authblk}
\usepackage{enumerate,color}

\newtheorem{proposition}{Proposition}

\newcounter{hypA}
\newenvironment{hypA}{\refstepcounter{hypA}\begin{itemize}
  \item[({\bf A\arabic{hypA}})]}{\end{itemize}}
  
\begin{document}
\title{The Time Machine: A Simulation Approach for Stochastic Trees} 
\author[1]{Ajay Jasra}
\author[2]{Maria De Iorio}
\author[2]{Marc Chadeau-Hyam}

\affil[1]{\small{Department of Mathematics, Imperial College London, London, SW7 2AZ, UK.}\newline
E-Mail:\emph{\texttt{a.jasra@ic.ac.uk}}}
\affil[2]{\small{School of Public Health, Imperial College London, London, W2 1PG, UK.}\newline
E-Mail:\emph{\texttt{m.deiorio@ic.ac.uk}},\emph{\texttt{m.chadeau@ic.ac.uk}}}

\date{}
\maketitle

\begin{abstract}
In the following paper we consider a simulation technique for stochastic trees.
One of the most important areas in computational genetics is the calculation
and subsequent maximization of the likelihood function associated to such
models. This typically consists of using importance sampling (IS) and sequential Monte Carlo (SMC) techniques. The approach proceeds by simulating the tree, backward in time from observed data, to a most recent common ancestor (MRCA).
However, in many cases, the computational time and variance
of estimators are often too high to make standard approaches useful. In this
paper we propose to stop the simulation, subsequently yielding biased
estimates of the likelihood surface. The bias is investigated from a theoretical point of view. Results from simulation studies are also given to investigate the balance between loss of accuracy, saving in computing time and variance reduction.
\end{abstract}

\pagestyle{myheadings}
\thispagestyle{plain}

\section{Introduction}
There is currently much interest in performing ancestral inference from molecular population genetic data. To facilitate this inference, there has been an explosion of research in developing computationally efficient methods. 
These techniques are designed either to compute the likelihood, for maximum likelihood estimation, of a sample of genes or for deriving the posterior distribution on parameters in coalescent models, which describe the ancestry of the genes. Broadly speaking there are three main approaches to inference in molecular population genetics: (i) importance sampling for likelihood evaluation, whose application in population genetics was pioneered  by
(Griffiths \& Tavar\'e, 1994a,b,c) (ii) Markov chain Monte Carlo methods (e.g. Kuhner et al.~(1995), Wilson \& Balding (1998)) (iii) Approximate Bayesian Computation (ABC) (Del Moral et al.~(2009), Marjoram et al.~(2003)). See Stephens (2004) for a review.

In this paper we concentrate on likelihood-based methods. 
Molecular data  have a sampling distribution which is a mixture over possible ancestries. The state space of the ancestries is huge and closed-form expressions are available only in the simplest cases. The objective is to calculate a parameter $\theta\in\Theta\subseteq\mathbb{R}^{d_{\theta}}$ ($d_{\theta}\in\mathbb{Z}^+$)
such that
\begin{equation}
l(y;\theta^*) := \sup_{\theta\in\Theta} \int_{F}l^c(z,y;\theta)\pi_{\theta}(z)dz
\label{eq:like}
\end{equation}
for some observed genetic data $y\in E$, parameter $\theta\in\Theta$, probability density $\pi_{\theta}$ on $F$ and $l^c:E\times F\times
\Theta\rightarrow\mathbb{R}^+$ an integrable function. 
Note that $y$ is typically the genetic types of a random sample of chromosomes.
In addition, $z=(z_0,\dots,z_k)$ denotes the coalescent history, i.e.~the set of ancestral configurations at the embedded events in a Markov process where coalescence, mutations or other events take place. $z_k$ denotes the current state, while $z_0$ is the state when a singleton ancestor is reached.
 
Statistical inference associated to $l(y;\theta)$ can be regarded as a missing data problem and could, in principle,
be tackled by the EM algorithm (Dempster et al.~1977) and its Monte Carlo extensions (e.g.~Fort \& Moulin\'es (2003)).
However, $z$, the stochastic tree, can be computationally expensive to simulate and such techniques are typically avoided. For example, for the coalescent (Kingman, 1982) and ancestral
recombination graphs (e.g.~Fearnhead \& Donelly (2001)), the standard approach is to use
IS (De Iorio \& Griffiths, 2004a; Griffiths \& Tavar\'e, 1994a; Stephens \& Donelly, 2000) and SMC methods (Chen et al.~2005)
to approximate \eqref{eq:like}. 
These approximations are usually computed on a discrete grid $\Delta_{\theta}\subseteq\Theta$
and the estimate of $\theta^*$ corresponds to the largest approximated likelihood on $\Delta_{\theta}$. See also Olsson \& Ryd\'en (2008) for an alternative procedure for state-space models.

Techniques such as ABC 
and composite likelihood (Wiuf, 2006) 
do not give solutions which are exact w.r.t.~the original model whilst, when possible, exact inference is of interest. This is because, given a reasonable stochastic
model, the approach allows investigators to exactly (up-to a numerical error) average over the uncertainty
in the tree structure when estimating genetic parameters of interest.
One of the main drawbacks of existing exact IS/SMC schemes
is the simulation of the tree backward in time, from observed data, until the tree coalesces.
In many scenarios, especially for large data sets, when getting close to the top of the tree, it often takes a long time to coalesce. This is due to genetic parameters (e.g. mutation rates) that can be very large relative  to the size of the data. Consequently, it can take a very long time to simulate the tree back to the MRCA.
As a result, the variance
of the estimate of the likelihood can be higher than is desirable, along with long CPU times.
It should be noted that the calculation of the likelihood at these points, $\theta\in\Theta$, can be inferentially
important. In addition, it is seldom possible to speed up the simulation via
importance sampling as the variance of the weights can become too large. That is, by adapting the parameter of the proposal to lead to a fast coalescence, the discrepancy between the true process and the proposal leads to a very inefficient algorithm w.r.t.~variance.

\subsection{The Time Machine}

The approach proposed in this paper is based on IS. Stephens \& Donelly (2000) proposed a way to use IS efficiently to simulate ancestral trees by characterizing an optimal proposal distribution and similar methods have since been developed for a variety of genetic scenarios (e.g.~De Iorio \& Griffiths (2004a,b)). The basic idea is to define an efficient proposal distribution on ancestral histories which allows us to reconstruct Markov histories backwards in time from the sample $y$ to an MRCA.  

We introduce a stopping time in the IS proposal, backward in time,
to stop the simulation before the MRCA is reached. Then using a simple stopped identity, forward in
time we are able to characterize the bias introduced in the evaluation of the likelihood due to stopping the simulation of the stochastic tree.
The bias can be understood by considering two aspects:
\begin{enumerate}
\item{The underlying mixing of the evolutionary process\label{pt:evpr}}
\item{The last exit time distributions on the process. \label{pt:lepr}}
\end{enumerate} 
In the context of (\ref{pt:evpr}), the idea is that for many models, close to the top
of the tree, the process is able to forget its initial condition. As a result,
stopping the simulation is reasonable, because the place where it is stopped is forgotten by the process forward in time; we formalize these ideas later on.
In reference to (\ref{pt:lepr}), the more information there is on the true
marginal distributions of the process, the more it is possible to reduce
the bias. Ideas from the theory of population genetics models
(Ethier \& Griffiths 1987; Ewens, 1972) will be used to achieve the latter.
 
In reference to a comment of Edwards (2000), our method is termed the `time machine'. This is because, estimation is performed saving the simulation time of going all the way back in time to the MRCA.
A similar
idea, in the context of filtering, can be found in the work of Olsson et al.~(2008) and also in option pricing Avramidis \& L\'Ecuyer (2006).
In our context, we have a simpler underlying process than in filtering, but the ergodicity
conditions considered there do not apply here. The mixing conditions that they require only
apply locally and thus the proofs have to be modified.
Recall that approximate tools for inference from stochastic
trees (e.g.~Del Moral et al.~(2009), Meligkotsidou \& Fearnhead (2007), Tavar\'e et al.~(2000)) are available. However,
our approach is `less approximate', in that our 
point-wise estimate of the likelihood is significantly less-biased, but costing more in computational-time.

This paper is structured as follows. In Section \ref{sec:motex} we introduce a motivating example, the coalescent model, which will help to illustrate our ideas. In
Section \ref{sec:stopsimos} our methodology is described; Section \ref{sec:bias} features an analysis of the bias of the approach;
Section \ref{sec:simos} presents a simulation study to demonstrate the performance of our algorithm and we conclude the paper in Section \ref{sec:summary}.
Appendix 1 contains some proofs, Appendix 2 details of our numerical implementations.
Our ideas are illustrated in the context of the coalescent. However, the formulation is kept as general as possible, as the framework can be extended
to other tree models, such as the infinite sites model. In Appendix 3 we show how this can be done. 

\section{Motivating Example}\label{sec:motex}

The coalescent model is used as a motivating example for our work.
Some notations are first introduced. In particular, we consider the case in which the type space $E=\{1,\ldots,d\}^n$ for the collection of the $n\in\mathbb{Z}^+$ genes/chromosomes is finite and the only genetic process of interest is mutation.  

\subsection{Notation}\label{sec:motivation}

Denote by $(E,\mathscr{E})$ a measurable space. 
For two $\sigma-$finite measures $\lambda_1$ and $\lambda_2$ mutual absolute continuity
is written $\lambda_1\sim\lambda_2$ and the Radon-Nikodym derivative as
$d\lambda_1/d\lambda_2$.
Given a Markov kernel
$P:E\times\mathscr{E}\rightarrow[0,1]$, let
$P^0(x,\cdot)=\delta_x(\cdot)$, (the Dirac measure)
and write the composition for $j\geq 1$
as $P^j(x,\cdot) = \int P(x,dy) P^{j-1}(y,\cdot)$, with a corresponding composition of inhomogeneous kernels as $P_{1:j}$. Write $\mathbb{I}_A$ as the indicator of a set. 
For $\text{Card}(E)<\infty$ 
$$
\mathcal{S}(E)=\{P=(p_{ij})_{i,j\in E}:p_{ij}\geq 0, \sum_{l\in E}p_{il}=1\cap
\exists 
\psi_i\geq 0, \forall i\in E, \sum_{l\in E} \psi_l = 1, 
\psi P = \psi \}
$$ 
denotes the class of stochastic matrices
for which there exist a stationary distribution $\psi$.
The collection of bounded and measurable function are denoted
$\mathcal{B}_b(E)$. The supremum norm is written $\|f\|_{\infty}
=\sup_{x\in E}|f(x)|$. The total variation distance between
two probability measures $\lambda_1$ and $\lambda_2$ on $(E,\mathscr{E})$ is $\|\lambda_1-\lambda_2\|_{tv}
:=\sup_{A\in\mathscr{E}}|\lambda_1(A)-\lambda_2(A)|$.
Given a probability
measure $\lambda$, and a $j\in\mathbb{Z}^+$, the product measure is written $\lambda^{\otimes j}:= \lambda^{\otimes j-1}\times\lambda$, $\lambda^{\otimes -1}:=1$.
The vector notation $x=(x_1,\dots,x_j)=x_{1:j}$ is adopted.
 In addition, let the $d$-dimensional vector $e_i=(0,\dots,0,1,0,\dots,0)$
where the 1 is in the $i^{th}$ position. The $\mathbb{L}_1$ norm of a vector
is written $|x_{1:j}|_1:=|x_1|+\cdots+|x_j|$.
For $d\in\mathbb{Z}^+$, $\mathbb{T}_d=\{1,\dots,d\}$.

\subsection{Identity of Interest}\label{sec:id_interest}
Define the tree model on the measurable space $(F,\mathscr{F})$,
with $\mathscr{F}=\sigma(F)$.
Let $n\geq 2$. The basic idea is to maximize, w.r.t $\theta\in\Theta$, the quantity
\begin{equation}
l(y_{1:n};\theta) = 
\sum_{k\in\mathcal{K}_n}\int_{F_{n}^k} \pi_{\theta}(z_{1:k})\mathbb{I}_{\{y_{1:n}\}\times
B_{n+1}}\big(z_{k-1},t(z_{k})\big)l^c(y_{1:n},z_{k-1};\theta)dz_{1:k}
\label{eq:likelihood_id}.
\end{equation}
where the observed data is $y_{1:n}\in E$, $t:F_n\rightarrow\mathbb{Z}^{d_t}$,
normally the identity, for $m\in\mathbb{Z}^+$
$$
B_m = \{x\in\mathbb{Z}^{d_t}:|x_{1:d}|_1=m\}.
$$
and
$F=\bigcup_{k\in\mathcal{K}_{n}}\big(\{k\}\times F_n^k\big)$ for some $E\subset
F_n$
and $\mathcal{K}_n\subset\mathbb{Z}^+$
depending upon the model under study.
In all of our examples, $\pi_{\theta}(z_{1:k})$ corresponds to the density of a non-decreasing (in some sense) Markov process in discrete time, stopped at a random time $k\in\mathcal{K}_n$; that is
$$
\pi_{\theta}(z_{1:k}) = p_{\theta}(z_1)
\bigg\{\prod_{j=2}^k p_{\theta}(z_{j-1},z_{j})\bigg\}\mathbb{I}_{B_{n+1}}\{t(z_{k})\}.
$$
Throughout the article it is assumed that 
$\sum_{k\in\mathcal{K}_n}\int_{F_{n}^k}\pi_{\theta}(z_{1:k})=1$, i.e.~that
the stopping time is a.s.~finite w.r.t $\pi_{\theta}$. The stopping time will be determined by the first time that the tree is of `size' $n+1$.

Introduce an absolutely continuous distribution $Q_{\theta}$ on $F$ and
sample $(z_{1:k^{(i)}}^{(i)})_{1\leq i \leq N}$ according to $Q_{\theta}$, then the IS estimator of
$l(y;\theta)$ is
$$
S^N(\widetilde{l}_{\theta}) = \frac{1}{N}\sum_{i=1}^N 
\bigg[\frac{\pi_{\theta}(z_{1:k^{(i)}}^{(i)})\mathbb{I}_{\{y\}\times
B_{n+1}}(z_{k^{(i)}-1:k^{(i)}}^{(i)})l^c(y_{1:n},z_{k^{(i)}-1}^{(i)};\theta)}{Q_{\theta}(z_{1:k^{(i)}}^{(i)})}\bigg]
$$
where
$$
\widetilde{l}_{n,\theta}(z_{1:k},k) = \frac{\pi_{\theta}(z_{1:k})\mathbb{I}_{\{y\}\times
B_{n+1}}(z_{k-1:k})l^c(y_{1:n},z_{k-1};\theta)}{Q_{\theta}(z_{1:k})}
$$
and 
$$
S^N(\cdot)=\frac{1}{N}\sum_{i=1}^N \delta_{z_{1:k^{(i)}}^{(i)}}(\cdot)
$$ 
the empirical measure of the simulated samples. 

\subsection{The Coalescent Model} 
Denote the number of genes of type $i$ at event $j$ of the process
as $z_j^i$, with $z_{j}=(z_j^1,\dots,z_{j}^d)$.
The objective is to find the genetic parameters $\theta=(\mu,P)$
where $\mu\in\mathbb{R}^+$ and $P\in\mathcal{S}(\mathbb{T}_d)$,
$\Theta=\mathbb{R}^+\times \mathcal{S}(\mathbb{T}_d)$. $\mu$ is the mutation rate per chromosome per generation and mutations along the edges of the tree occur according to a Markov chain with transition matrix $P$. 

The various components of the identity (\ref{eq:likelihood_id}) for the coalescent model are defined as:
\begin{eqnarray*}
F & = &\bigcup_{k\in\mathcal{K}_{n}}\bigg(\{k\}\times F_n^{k}\bigg)\\
F_n & = &\{z^{1:d}\in(\mathbb{Z}^+\cup\{0\})^d : 2\leq |z^{1:d}|_1\leq n+1\}\\
\mathcal{K}_{n} & = & \{n, n+1, \dots\}
\end{eqnarray*}
with $t$ the identity function, 
$$
l^c(y_{1:d}^n,z;\theta) = \left\{ \begin{array}{ll}
\frac{\prod_{j=1}^d y_j^n!}{n!} & \textrm{if}\quad y_{1:d}^n=z\\
0 & \textrm{otherwise}
\end{array}\right.
$$
and finally,
$$
\pi_{\theta}(z_{1:k}) = \mathbb{I}_{\{z:|z|_1=n+1\}}(z_k)\}\bigg\{\prod_{j=2}^k p_{\theta}(z_{j-1},z_j)
\bigg\}\int p_{\theta}(z_0)p_{\theta}^1(z_0,z_1)dz_0
$$
where
$$
p_{\theta}(z_{j-1},z_j) = \left\{\begin{array}{ll}
\frac{z_{j-1}^i}{|z_{j-1}|_1}\frac{\mu}{|z_{j-1}|_1-1+\mu}p_{il} & \textrm{if}\quad z_j =
z_{j-1} - e_i+e_l\\ 
\frac{z_{j-1}^i}{|z_{j-1}|_1}\frac{|z_{j-1}|_1-1}{|z_{j-1}|_1-1 + \mu} & \textrm{if}\quad z_j = z_{j-1} + e_i\\ 
0 & \textrm{otherwise}.
\end{array}\right.
$$
$p_{\theta}^1(z_0,z_1)=\mathbb{I}_{\{z:z=2z_0\}}(z_1)$ and
$$
p_{\theta}(z_0) = \left\{\begin{array}{ll}
\psi_{\theta}(i) & \textrm{if}\quad z_0=e_i\\
0 & \textrm{otherwise}.
\end{array}\right.
$$
Write $p_{\theta}^1(z_1)=\int p_{\theta}(z_0)p_{\theta}^1(z_0,z_1)dz_0$ (here $dz$ is counting measure).
Note that for any fixed $n\geq 2$, $\theta\in\Theta$, $\sum_{k\in\mathcal{K}_n}\int_{F_n^k} \pi_{\theta}(z_{1:k})dz_{1:k} = 1$.
For simplicity of exposition, the results are given with only mutation.
However, they can be easily extended to the case of migration as well (e.g.~De Iorio \& Griffiths (2004b)).

\subsection{Likelihood Computation}\label{sec:likelihood_computation}

To compute the likelihood, for a given $\theta\in\Theta$, importance sampling is adopted. 
An importance distribution, $Q_{\theta}$, is introduced to simulate the tree backward in time to the MRCA; this ensures that the data is hit. 

In details, let $x$ denote the reverse chain backward in time and write $x\in F_n$ instead of $z$ (this convention is used throughout the article, see also Figure \ref{fig:coalgraph}). Let:
$$
Q_{\theta}^{y_{1:d}^n}(x_{1:k-1}) =  \mathbb{I}_{\{y_{1:d}^n\}}(x_{1})\bigg\{\prod_{j=2}^{k-1}
q_{\theta}(x_{j-1},x_{j})\bigg\} \mathbb{I}_{\{x\in(\mathbb{Z}^{+}\cup\{0\})^d: x = e_i, i\in\mathbb{T}_d\}}(x_{k-1})
$$
for some Markov transition $q_{\theta}$; see Stephens \& Donelly (2000) for the optimal $Q_{\theta}$.
Then the likelihood is
\begin{eqnarray*}
l(y_{1:d}^n;\theta) & = & \frac{n-1}{n-1+\mu}\frac{\prod_{j=1}^d (y_j^n)!}{n!}\sum_{k\in\mathcal{K}_{n-1}}
\int_{F_{n-1}^{k-1}}p_{\theta}(x_{k-1})\bigg\{\prod_{j=2}^{k-1}\frac{p_{\theta}(x_{j},x_{j-1})}{
q_{\theta}(x_{j-1},x_{j})}\bigg\}\times 
\\ & & 
Q_{\theta}^{y_{1:d}^n}(x_{1:k-1})dx_{1:k-1}.
\end{eqnarray*}
The simulation proceeds by sampling from $q_{\theta}(y_{1:d}^n,\cdot)$
and computing the weight
$$
w_{\theta}(y_{1:d}^n,x_2) = \frac{p_{\theta}(x_{2},y_{1:d}^n)}{q_{\theta}(y_{1:d}^n,x_{2})}
$$
Simulations backward in time are carried out until we reach the MRCA, i.e. when there is only one individual in the sample. This procedure is repeated $N$ times to provide
a Monte Carlo estimator for the likelihood
$$
S^N(\widetilde{l}_{n,\theta})
=
\bigg\{\frac{n-1}{n-1+\mu}\frac{\prod_{j=1}^d (y_j^n)!}{n!}\bigg\}
\frac{1}{N}\sum_{i=1}^N
\bigg[p_{\theta}(x_{k^{(i)}-1}^{(i)})\bigg\{\prod_{j=2}^{k^{(i)}-1}
w_{\theta}(x_{j-1}^{(i)},x_j^{(i)})\bigg\}\bigg]
$$
where $(x_{1:k^{(i)}-1}^{(i)})_{1\leq i \leq N}$ are the simulated samples,
$x_{1}^{(i)}=y^{n}_{1:d}$ for every $i\in\mathbb{T}_N$
and
$$
\widetilde{l}_{n,\theta}(x_{1:k-1},k) = \bigg\{\frac{n-1}{n-1+\mu}\frac{\prod_{j=1}^d (y_j^n)!}{n!}\bigg\}
p_{\theta}(x_{k-1})
\bigg\{\prod_{j=2}^{k-1}
w_{\theta}(x_{j-1},x_j)\bigg\}.
$$
This can be repeated for many $\theta$ using a driving value (Griffiths \& Tavar\'e, 1994) or bridge sampling
ideas (e.g.~Fearnhead \& Donelly (2001)). In addition, to deal with the problem of weight degeneracy (e.g.~Doucet et al.~(2001)) resampling steps can be added. See, for example, Chen et al.~(2005).

\begin{figure}[h]\centering
\setlength{\unitlength}{8cm}
\begin{picture}(1,1)
\put(.5,.75){\line(0,1){0.25}}
\put(0.25,.75){\line(1,0){.5}}
\put(0.75,.75){\line(0,-1){0.5}}
\put(0.25,.75){\line(0,-1){0.1}}
\put(0.125,.65){\line(1,0){0.25}}
\put(0.125,.65){\line(0,-1){0.65}}
\put(0.375,.65){\line(0,-1){0.15}}
\put(0.3125,.5){\line(1,0){0.125}}
\put(0.3125,.5){\line(0,-1){0.5}}
\put(0.4375,.5){\line(0,-1){0.5}}
\put(0.625,.25){\line(1,0){.25}}
\put(0.625,.25){\line(0,-1){.25}}
\put(0.875,.25){\line(0,-1){.25}}
\multiput(0.1,0.5)(0.1,0){8}{\line(1,0){0.05}}
\put(0.95,.95){$\textrm{Forward}$}
\put(0.01,.95){$\textrm{Backward}$}
\put(0.95,.75){$Z_1$}
\put(0.95,.65){$Z_2$}
\put(0.95,.5){$Z_3$}
\put(0.95,.25){$Z_4$}
\put(0.95,.01){$Z_5$}
\put(0.01,.5){$X_3$}
\put(0.01,.25){$X_2$}
\put(0.01,.01){$y_{1:d}^5$}
\put(0.25,.45){$\rho$}
\put(0.7,.45){$\alpha-1$}
\end{picture}
\caption{A Coalescent tree. Here there are no mutations on the tree, and
$B_{3}$ is the set used to stop the simulation, backward in time (see Section \ref{sec:stopcoal}).
The horizontal dotted line represents the random times $\rho$ and $\alpha-1$.}
\label{fig:coalgraph}
\end{figure}
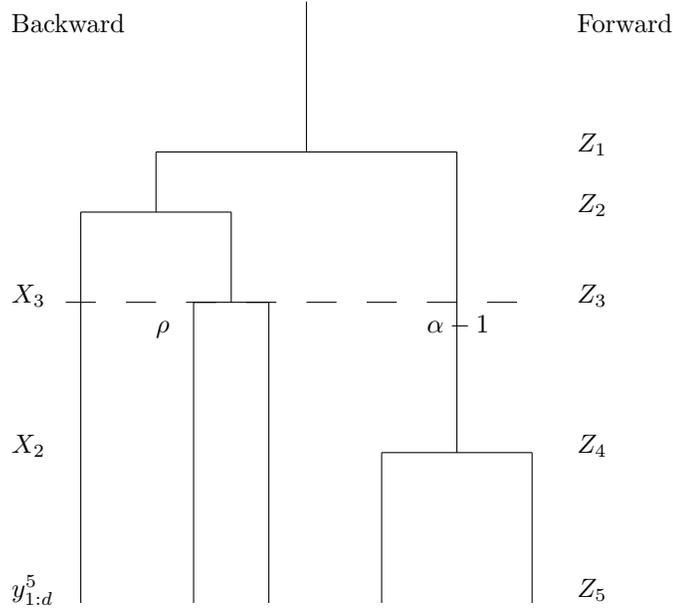

\section{Stopping the Simulation}\label{sec:stopsimos}

It is now detailed how we stop the simulation of the stochastic tree back in time before the MRCA is reached. In the next Section we provide theoretical results and connections to the theory of SMC are established. For the purpose of stopping the simulation, introduce two stopping
times (forwards in time): the first hitting time of the set $B_{n+1}$
$$
\tau := \inf\{k\geq 1:|t(Z_k)|_1=n+1\}
$$
and some stopping time $\alpha$ associated to the hitting of a set $A\in\mathscr{F}$
$$
\alpha := \inf\{k\geq 1:Z_k\in A\}
$$
such that
$$
\mathbb{P}_{\pi_{\theta}}(\alpha < \tau) = 1
$$
where $\mathbb{P}_{\pi_{\theta}}$ is the $\pi_{\theta}-$probability.
For example, in the context of the coalescent, it is suggested to take,
for $m<n$
$$
\alpha = \inf\{k\geq 1:|Z_k|_1=m+1\}. 
$$

\subsection{A Stopped Identity}
Let $\mathbb{E}_{\pi_{\theta}}$ denote the expectations 
w.r.t the process $\{Z_k\}$. Then 
the likelihood (\ref{eq:likelihood_id}) can be written as
$$
l(y_{1:n};\theta) = 
\mathbb{E}_{\pi_{\theta}}\bigg[\mathbb{I}_{\{y_{1:n}\}\times B_{n+1}}\{Z_{\tau-1},t(Z_{\tau})\}
l^c(y_{1:n},Z_{\tau-1};\theta)\bigg]
$$
and applying the strong Markov property we have
$$
l(y_{1:n};\theta) = \mathbb{E}_{\pi_{\theta}}\bigg[
\mathbb{E}\bigg(\mathbb{I}_{\{y_{1:n}\}\times B_{n+1}}\{Z_{\tau-1},t(Z_{\tau})\}l^c(y_{1:d}^n,Z_{\tau-1};\theta)\bigg|Z_{\alpha}\bigg)\bigg]\;
$$
that is,
\begin{eqnarray}
l(y_{1:n};\theta) & = &
\sum_{\alpha}\int \pi^u_{\theta}(z_{\alpha})
\bigg[\sum_{\tau}\int \bigg\{\prod_{i=\alpha+1}^{\tau}p_{\theta}(z_{i-1},z_i)\bigg\} l^c(y_{1:d}^n,z_{\tau-1};\theta)
\times \nonumber\\ & & 
\mathbb{I}_{\{y_{1:n}\}\times B_{n+1}}\{z_{\tau-1},t(z_{\tau})\}
dz_{\alpha+1:\tau}\bigg] dz_{\alpha}
\label{eq:stoplike}
\end{eqnarray}
where
$$
\pi^u_{\theta}(z_{\alpha}) = \int p_{\theta}^1(z_1)\bigg\{\prod_{i=2}^{\alpha}p_{\theta}(z_{i-1},z_i) \bigg\}\mathbb{I}_{(A^c)^{\alpha-1}\times A}(z_{1:\alpha})dz_{1:\alpha-1}.
$$
The equation (\ref{eq:stoplike}) will be the starting point
for constructing our biased estimates of the likelihood function.

\subsection{Coalescent Model}\label{sec:stopcoal}

Consider the coalescent model.
Specifically, define, for $n>m\geq 3$
the stopping time $\alpha$
$$
\alpha = \inf\{k\geq 1: Z_k\in B_{m+1}\}
$$
which is the first time the forward process has $m+1$ individuals.

Using equation (\ref{eq:stoplike}), we have
\begin{equation}
l(y_{1:n};\theta) = \sum_{\alpha=m}^{\infty}\sum_{\tau=\alpha + n - m}^{\infty}
\int \pi^u_{\theta}(z_{\alpha})
\bigg[\int \bigg\{\prod_{i=\alpha+1}^{\tau}p_{\theta}(z_{i-1},z_i)\bigg\} l^c(y_{1:d}^n,z_{\tau-1};\theta)dz_{\alpha+1:\tau}\bigg]dz_{\alpha}
\label{eq:coallikestp}.
\end{equation}
In words this means that to have $m+1$ chromosomes, we need a minimum of $m$ steps in the
process and $\tau$ has to be at least $n-m+\alpha$ steps. 

In this case, write
$$
\pi^u_{\theta}(z_{\alpha}) = \int \psi_{\theta}(z_1)
\bigg\{\prod_{i=2}^{\alpha}p_{\theta}(z_{i-1},z_i)\bigg\} 
\mathbb{I}_{B_m\times B_{m+1}}(z_{\alpha-1:\alpha})dz_{1:\alpha-1}\;
$$
Note this is well-defined due to the fact that the size of the population
is non-decreasing, and then, for any $\alpha\in\mathcal{K}_m$
$$
\pi^u_{\theta}(z_{\alpha}) = 
\int 
\mathbb{I}_{B_{m+1}}(z_{\alpha})
\pi^e_{\theta}(z_{\alpha-1})p_{\theta}(z_{\alpha-1},z_{\alpha})
dz_{\alpha-1}
$$
where
$$
\pi^e_{\theta}(z_{\alpha-1}) = \frac{m-1+\mu}{m-1}\int p^1_{\theta}(z_1)\bigg\{\prod_{i=2}^{\alpha-1}p_{\theta}(z_{i-1},z_i) \bigg\}p_{\theta}(z_{\alpha-1},u)
\mathbb{I}_{B_m\times B_{m+1}}(z_{\alpha-1},u)dz_{1:\alpha-2}du.
$$
That is, given $\alpha$, the distribution of the chromosome counts at the first entrance time of $B_{m+1}$ can be written as the composition of: 
\begin{itemize}
\item{the distribution of the counts at the last exit time from $B_{m}$}
\item{and the Markov transition.} 
\end{itemize}

Returning to the likelihood (\ref{eq:coallikestp}) and making the substitutions,
$\alpha'=\alpha-1$,
$\eta = \tau - \alpha' +1$, it thus follows that
\begin{eqnarray*}
l(y_{1:n};\theta) & = & 
\sum_{\eta=n-m+2}^{\infty}\sum_{\alpha'=m-1}^{\infty}
\int \pi^e_{\theta}(z_{\alpha'})\mathbb{I}_{B_{m+1}}(z_{\alpha'+1})
\bigg\{\prod_{i=\alpha'+1}^{\eta + \alpha' - 1}
p_{\theta}(z_{i-1},z_i)\bigg\}\times
\\ & & 
\mathbb{I}_{\{y_{1:d}^n\}\times B_{n+1}}(z_{\eta+\alpha'-2:\eta+\alpha'-1})
l^c(y_{1:d}^n,z_{\eta+\alpha'-2};\theta)
dz_{\alpha':\eta+\alpha'-1}.
\end{eqnarray*}
Here $\eta$ is the time from the last time there are $m$ chromosomes
to $n+1$ chromosomes.
Now set
$$
\rho = \inf\{k\geq 1: |X_k|\in B_m\}.
$$
In other words the simulation is stopped the first time there are $m$ chromosomes.
Our approximation of the likelihood is then
$$
l_b(y_{1:n};\theta) =
\sum_{\rho=n-m+2}^{\infty} \int h_{\theta}(x_{\rho})
\bigg\{\prod_{i=2}^{\rho}p_{\theta}(x_{i},x_{i-1})\bigg\}\mathbb{I}_{\{y_{1:d}^n\}\times
B_{n+1}}(x_{2:1})l^c(y_{1:d}^n,x_2;\theta)
dx_{1:\rho}.
$$
On the basis of the above analysis, it is then clear that if 
\begin{equation}
h_{\theta}(x) = \sum_{\alpha=m-1}^{\infty} \int \pi^e_{\theta}(z_{\alpha})\mathbb{I}_{\{z_{\alpha}\}}(x)dz_{\alpha}
\label{eq:lastexit}
\end{equation}
then the approximation of the likelihood is exact. That is, to minimize
the bias an approximation of the true distribution of the counts at the last time there are $m$ chromosomes should be used.
The ideas and notation are clarified in Figure \ref{fig:coalgraph}.

\section{Results on the Bias}\label{sec:bias}

In our biased simulation, using the decomposition \eqref{eq:stoplike}, the procedure will approximate
$$
l_b(y_{1:n};\theta) =
\sum_{\rho}\int h_{\theta}(x_{\rho})
\bigg\{\prod_{i=2}^{\rho}p_{\theta}(x_{i},x_{i-1})\bigg\}\mathbb{I}_{\{y_{1:d}^n\}\times
B_{n+1}}(x_{2:1})l^c(y_{1:d}^n,x_{2};\theta)
dx_{1:\rho}\;
$$
where our notation is such that:
\begin{itemize}
\item{$\big(X_k\big)_{k\geq 1}$ is the time reversed process}
\item{$\rho$ is a first hitting time associated to $\big(X_k\big)_{k\geq 1}$}
\item{$h_{\theta}$ an approximation
of a marginal probability.} 
\end{itemize}

\subsection{Error Bounds}

We begin by giving a simple result on the error bounds for SMC algorithms. The result applies to the standard IS algorithms, for example in De Iorio \& Griffiths (2004b), Stephens \& Donelly (2000),
and for the SMC algorithms as in Chen et al.~(2005). The simulation is to be performed
backward in time, as in Section \ref{sec:likelihood_computation}.
The ideas here are adapted from the theory of Del Moral (2004).

The biased estimates are denoted as $S_b^N(\widetilde{l}_{n,\theta})$, $N\geq 1$, where
$\widetilde{l}_{n,\theta}$ depends upon whether IS or SMC is implemented.
For example, in the IS case:
$$
\widetilde{l}_{n,\theta}(x_{1:\rho},\rho) = 
\frac{h_{\theta}(x_{\rho})
\bigg\{\prod_{i=2}^{\rho}p_{\theta}(x_{i},x_{i-1})\bigg\}\mathbb{I}_{\{y_{1:d}^n\}\times
B_{n+1}}(x_{2:1})l^c(y_{1:d}^n,x_{2};\theta)}{Q_{\theta}(x_{1:\rho})}
$$
where
$$
Q_{\theta}(x_{1:\rho}) = \mathbb{I}_{\{y_{1:d}^n\}\times
B_{n+1}}(x_{2:1})\bigg\{\prod_{i=2}^{\rho-1}q_{\theta}(x_{i},x_{i+1})\bigg\} \mathbb{I}_{(A^{c})^{\rho-1}\times A}(x_{1:\rho})
$$
is such that $A$ is the set associated to $\rho\geq 3$ ($Q_{\theta}-$a.s.),
$x_1,x_2\notin A$
and $\sum_{\rho}\int Q_{\theta}(x_{1:\rho})dx_{1:\rho}=1$.
Below expectations w.r.t the stochastic process that is simulated by the algorithm are written as $\mathbb{E}$ and it is assumed
$$
\|\widetilde{l}_{n,\theta}(x_{1:\rho},\rho)\|_{\infty}<+\infty \quad \forall \theta\in\Theta.
$$ 

\begin{proposition}\label{prop:lpbound}
For
any $n\geq 2$, $p\geq 1$, $\theta\in\Theta$, $y_{1:n}$, there exists a  $B_{p,n}(\theta)<+\infty$ such that:
$$
\mathbb{E}[|S_b^N(\widetilde{l}_{n,\theta})-l(y_{1:n};\theta)|^p]^{1/p} \leq \frac{B_{p,n}(\theta)}{\sqrt{N}}
+ |l_b(y_{1:n};\theta) - l(y_{1:n};\theta)|.
$$
\end{proposition}

\noindent\textbf{Remark}.\emph{
The result shows the standard variance-bias type decomposition. 
That is, $B_{p,n}(\theta)/\sqrt{N}$ can be thought of as a bound on the variance and 
$|l_b(y_{1:n};\theta) - l(y_{1:n};\theta)|$ is the bias.
Our estimate converges to $l_b(y_{1:n};\theta)$, and it is sought to control the
bias term, which, in our case can be approximately written in the form
\begin{equation}
|l_b(y_{1:n};\theta) - l(y_{1:n};\theta)|
= |[\lambda_1-\lambda_2](P_{1:k}(f))| \label{eq:biasdecomp}
\end{equation}
for $\lambda_1,\lambda_2$ two probability measures and $\{P_{n}\}$ a sequence of non-homogenous
Markov kernels ($\theta$ is suppressed on the R.H.S).
}

\subsection{Controlling the Bias}\label{sec:contrbias}

A simple technical result is now given which shows how to control
the bias term (\ref{eq:biasdecomp}). 

Some assumptions are now made, that can be
satisfied by many stochastic tree models. Introduce a sequence of time inhomogeneous Markov kernels $\{P_n\}$,
on space $(R,\mathscr{R})$ and a sequence of sets $\big(\{C_n:n\geq
0, C_n\in\mathscr{R}\}\big)_{n\geq 0}$.

\begin{hypA} \label{hyp:p_nassump}
{\em Stability of $\{P_n\}$.}
\renewcommand{\labelitemii}{}
\begin{itemize}
\item{(i) \bf Initial Probability Measures}. $\lambda_1,\lambda_2$ are concentrated
on $C_0$.
\item{(ii) \bf Absorption of $\{P_n\}$}. For every $n\geq 1$, $x\in C_{n-1}$ we have
\begin{equation}
P_n(x,C_{n})=1\label{eq:absorbcond}.
\end{equation}
\item{(iii) \bf Local Mixing of $\{P_n\}$}. For every $n\geq 1$, there exist $\epsilon_n\in(0,1)$,
$\nu_n$ concentrated on $C_{n}$, such that for all $x\in C_{n-1}$
\begin{equation}
\epsilon_n \nu_n(\cdot) \leq P_n(x,\cdot) \leq \frac{1}{\epsilon_n}\nu_n(\cdot)
\label{eq:mixingcond}.
\end{equation}
\end{itemize}
\end{hypA}

The assumption (A\ref{hyp:p_nassump}) (which is comprised of (i)-(iii)) will refer to the 
fast mixing of the process close to the top of the tree.
The absorption type assumption refers to the birth process associated
to coalescent type chains.

\begin{proposition}\label{prop:biascontrol}
Assume (A\ref{hyp:p_nassump}). Then, for any $k\geq 1$, define:
$$
\vartheta_k := \frac{2}{\epsilon_1^2\log 3}\prod_{i=2}^k\frac{1-\epsilon_i^2}{1+\epsilon_i^2}
$$
and we have
$$
\|[\lambda_1-\lambda_2]P_{1:k}\|_{tv} \leq \vartheta_k\|\lambda_1-\lambda_2\|_{tv}.
$$
\end{proposition}

\noindent\textbf{Remark 1}.\emph{
The result helps to bound the bias as
$$
|[\lambda_1-\lambda_2](P_{1:k}(f))| \leq \|f\|_{\infty}\|[\lambda_1-\lambda_2]P_{1:k}\|_{tv}.
$$
Essentially, the fast mixing of $P_n$ within the domain it is constrained to allow the composition of kernels to forget its initial distribution
at an exponential rate.
In addition, as in Olsson et al.~(2008), assuming $\epsilon_n$ is uniform
in $n$, the benefits of stopping, in terms of variance/bias trade off can be
substantial.
}

\noindent\textbf{Remark 2}.\emph{
One point of interest in the sequel is that, if the mixing condition
(A\ref{hyp:p_nassump}) does not hold, it is possible to establish
a similar bound when the initial measures $\lambda_1$ and $\lambda_2$ are similar.
That is to say, when $\lambda_1\sim\lambda_2$ and $\exists \epsilon\in(0,1)$ such that
$$
\epsilon\leq \frac{d\lambda_1}{d\lambda_2} \leq \frac{1}{\epsilon}.
$$
This is unsurprising as it implies that if the kernels do not mix, we need to `match'
$\lambda_1$ and $\lambda_2$ for the bias to be small.
}

\subsection{Verifying the Assumptions}

(A\ref{hyp:p_nassump}) is now discussed in the context of the coalescent. Note
that the results follow, with some extra work, for coalescent processes
with migration. Readers interested in how the method may be applied can skip to Section \ref{sec:simos}, with no loss in continuity.

Suppose that the transition matrix $P$ satisfies,
for any $i,j\in\{1,\dots,d\}$, $\epsilon_{\varphi}\in (0,1)$
and probability $\varphi$, $\varphi_j>0$
$$
\epsilon_{\varphi}\varphi_j \leq p_{ij} \leq
\epsilon_{\varphi}^{-1}\varphi_j.
$$
This condition implies that $P$ mixes extremely quickly.
Let $C_0=\{z_{1:d}:|z_{1:d}|_1 = 3\}$; this corresponds to
the space of $\lambda_1=\pi_{\theta}^e$. Also let $C_1=\{z_{1:d}:3\leq |z_{1:d}|_1\leq
6\}$.
It is clear that $p_{\theta}^3(x,C_1)=1$: since we start with at most 3
chromosomes and the most possible after 3 steps is 6. Now
it can be seen that, for any $z\in C_0$
$$
p_{\theta}^3(z,\cdot) \geq \epsilon_1 \varphi^{\otimes 3}(\cdot)
$$
and 
$$
p_{\theta}^3(z,\cdot) \leq \epsilon_1^{-1} \varphi^{\otimes 3}(\cdot)
$$
with
$$
\epsilon_1 = 6\bigg[3\epsilon_\varphi\big[\inf_{i\in\{1,\dots,d\}}\varphi(i)\big]\bigg(\frac{\mu}{\mu+2}\bigg)^2\bigg]^3.
$$
Here the minorising probability $\nu=\varphi^{\otimes 3}$ puts all its probability
on having 3 chromosomes. Then it can be subsequently seen that $p_{\theta}^6$
satisfies condition (\ref{eq:mixingcond}), with $C_2=\{z_{1:d}:3\leq |z_{1:d}|_1\leq 12, n_i\geq 0\}$ and so fourth. In effect 
the condition (\ref{eq:mixingcond}) holds with $\epsilon_n\rightarrow 0$; that
is, the closer to the top of the tree we stop, the faster the process
will mix forward in time.

As a result, to bound the bias we can write it, approximately, in
the form, for $r>n-m+1$
$$
M(r)\|\lambda_1-\lambda_2\|_{tv}\bigg[\sum_{k=n}^{r}\xi^{k/l} + R(r)\bigg]
$$
with $\lambda_1$ as in (\ref{eq:lastexit}), $\lambda_2=h_{\theta}$,
$M(r), R(r)\in(0,\infty)$,
$l\in\{1,\dots,k\}$ is associated to the fact that we need to iterate the
kernels to satisfy (\ref{eq:mixingcond}) and $\xi\in(0,1)$.
$r$ is an integer big enough (say $100000$) where we suspect that the possibility
of generating a tree of length $n$ and hitting the data is extremely small,
so we can neglect the upper term.
Thus, approximately, the bound shows that the bias falls geometrically as we stop closer to the top of the tree. Note, however, it cannot go to zero unless $\lambda_1$ and $\lambda_2$ are equal. To an extent, finding good approximations is more difficult than being able to stop the tree, which is why we focus on this.

\noindent\textbf{Remark 1}.\emph{
The result given here mirrors one
proved by Donelly \& Kurtz (1999) for Fleming-Viot models. In Theorem 9.4
of that paper they show that the particle process is uniformly ergodic, if
the mutation process is. This is very similar to the property established
above. }

\noindent\textbf{Remark 2}.\emph{The information, in terms of when to stop
the simulation, that is contained in the bound on the bias is as follows.
If the mutation process mixes quickly, as above, then the bias falls at a
geometric rate: we should stop the simulation when the process starts
to mutate many times. This could be measured in terms of the effective sample
size (e.g.~Liu (2001)), if trees are simulated in parallel, or alternatively, if $\frac{\mu}{M}\geq n$,
for $M\in\mathbb{Z}^+$ a large multiple of the current size of the tree.}

\noindent\textbf{Remark 3}. \emph{In terms of the expression
$\|\lambda_1-\lambda_2\|_{tv}$, one could adopt a parent-independent mutation (PIM) marginal.
If we have
$$
\sup_{x}\|p_{\theta}(x,\cdot) - p_{\theta,m}(x,\cdot)\|_{tv} \leq \sup_{i,j}|p_{ij}-\phi_j|
= \varrho
$$
where $p_{\theta,m}(x,\cdot)$ is the transition for the PIM, and the mutation
vector is $\phi_j$, then ideas from perturbed Markov chains (e.g.~Mitrophanov (2005)) can be adopted to determine a quantitative bound. We are currently investigating a meaningful bound. 
}

\section{Simulations}\label{sec:simos}

\subsection{Experiment Set-Up}

To illustrate our approach, we consider three simulation scenarios: two
PIM models and one parent dependent
mutation model (PDM). The two PIM models, denoted PIM 0.5-0.5 and PIM
0.1-0.9 are based on the following per-locus transition matrices:
$$
\begin{pmatrix}
0.5 &0.5\\
0.5 &0.5
\end{pmatrix}   
{\mbox{, and }}
\begin{pmatrix}
0.1 &0.9\\
0.1 &0.9
\end{pmatrix}   
{\mbox{ respectively,}}
$$
while the per-locus mutation probability matrix underlying the PDM
model is
$$
\begin{pmatrix}
0.5 &0.5\\
0.1 &0.9
\end{pmatrix}   
.
$$
In all three scenarios, the initial population was set to 100
sequences and we considered a single-locus case ($i.e.$ with 2
possible types). For the PDM model only, we also considered the case
of 10 loci ($i.e.$ $2^{10}$=1024 different types). Irrespective of the number of loci considered, the distribution of the 100 initial sequences among
the different types was sampled from a multinomial distribution with a
probability vector $P$ defined as the invariant point, solution of
equation, $P_{(2^n\times 1)}=P_{(2^n\times 1)}T_{(2^n\times 2^n)}$, 
where $n$ denotes the number of loci considered and $T$ is the full
mutation probability matrix. 

The algorithm description is given in Appendix 2. For the function $h_\theta$,
we use the distribution of an un-ordered sample from a PIM model (which, even for the PIM cases, is not the correct distribution in the bias term).

\subsection{Simulation Results}

Simulations were carried out until there
were $TM=1, 2, 5, 10, 25, 50$ sequences left in the population. $TM=1$
exactly corresponds the approach in Stephens \& Donelly (2000) and is
subsequently referred to as $SD$. For each simulation, we examined
60 values for $\mu$ ranging from 0.1
to 30.1. We report in Figure \ref{graphe_distrib} the estimated
log-likelihood distribution, based on 100,000 samples for all four
simulations scenarios (presented in lines) and three values of $\mu$
(presented in columns). 
\begin{figure}
\begin{center}
\includegraphics[width=\linewidth]{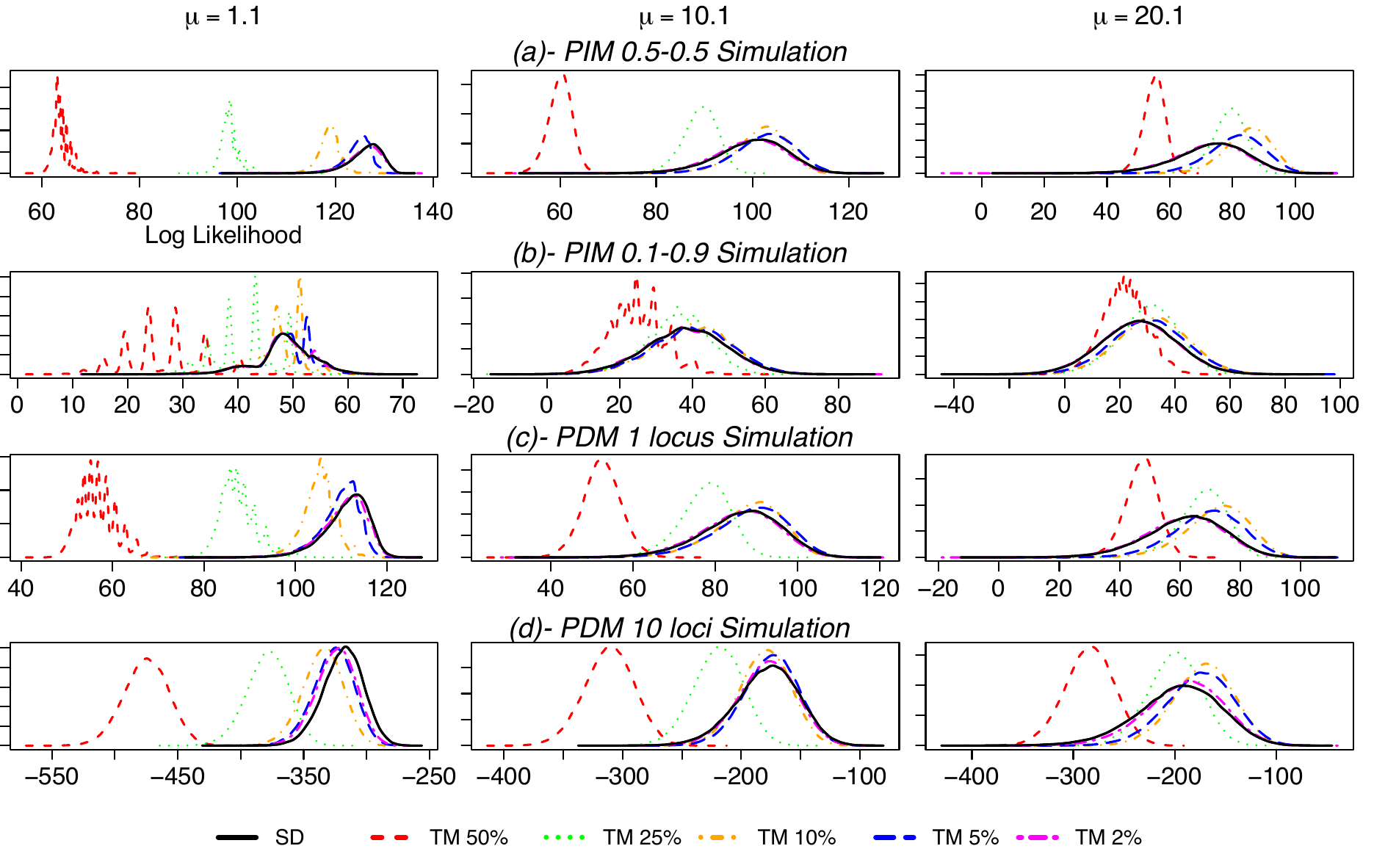}
\caption{\label{graphe_distrib} Estimated distribution of the
    likelihood for the four simulation scenarios and for three values
    the mutation rate $\mu=0.1, 10.1$ and $10.1$. In each model,
    results are presented for the six stopping times in the simulation
    of the genealogical tree. Plots are
  based on 100,000 samples. }
\end{center}
\end{figure}

In Figure \ref{graphe_distrib} it is clear that as expected, uniformly across the values of $\mu$, the
closer to the MRCA the algorithm is stopped, the more accurate the
distribution of the likelihood is estimated. However, up to
TM 10\% (and even TM 25\% for $\mu>20$) our results suggest
that the time machine approximation and correction
provides an accurate estimate of the distribution of the
likelihood. Conversely, when the algorithm is stopped too early
(TM $\geq$ 25\%) the biased estimator underlying the time machine
approach leads to very inaccurate estimates of the likelihood.
For even more extreme cases
(TM 50\% for $\mu=0.1$), this results in a highly shifted
estimated distribution of the likelihood. 

The above observations are also reflected in the mean likelihood (Figure \ref{meanL}). For every model considered here, the
simulations of the time machine up to TM 25\% seem to provide
estimates of the mean likelihood that are similar to
the SD approach, although for larger values of $\mu$, TM 25\%
seems to overestimate the mean likelihood. Furthermore, the time
machine approach seems to accurately locate the value of $\mu$
maximizing the likelihood for TM$\leq$ 10\%, and to provide acceptable approximations for for this when TM$>$ 10\%, regardless of the
simulation scenario.

\begin{figure}
\begin{center}
\includegraphics[width=\linewidth]{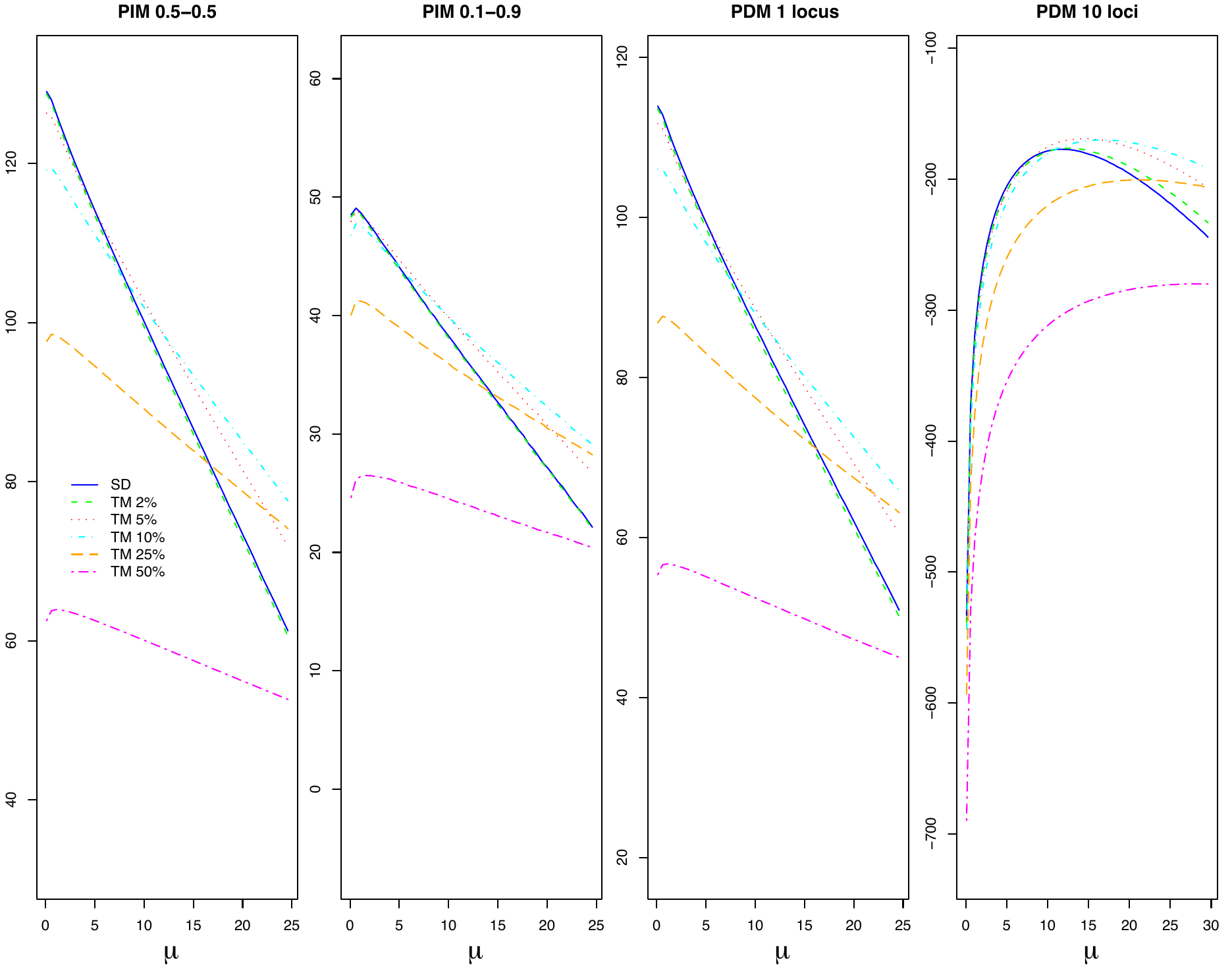}
\caption{\label{meanL} Estimated likelihood for the four
  simulation scenarios as a function of the mutation rate $\mu$. Plots are based on  100,000 samples.}
\end{center}
\end{figure}

In Figure \ref{meanT}, the average computation time per iteration is
plotted as a function of $\mu$ for the PDM-10 loci
simulations. Results for all other models led to the same conclusions and are therefore not shown.
From this figure, the computation time appears to be a linearly
increasing function of $\mu$: increasing the mutation rate naturally
decreases the probability of simulating a coalescent event and
therefore tends to increase the time 
to reach the MRCA (or any population size). However, it seems that
stopping the simulation when there are only more than 5 sequences left
in the population drastically reduces the computation time: for
TM 5\% the simulation run is on average more than twice as fast as
the SD simulation, and for TM 25\%, the time machine is more than
3 times more time-efficient than the SD algorithm.  It should also be noted that `large' values of $\mu$ (around 10), for which the time savings are most significant, also seem to be inferentially important (see the fourth panel of Figure \ref{meanL}).

In Figure \ref{fig:relative_sd} the relative standard deviation across our 100 repeats of the algorithm, of the time machine to SD are plotted for all the scenarios considered. It can be seen, as expected, that there is some variance reduction and, for example for the TM 5\% PDM, the variance reduction is of the order 1.5.

\begin{figure}
\begin{center}
\includegraphics[width=\linewidth]{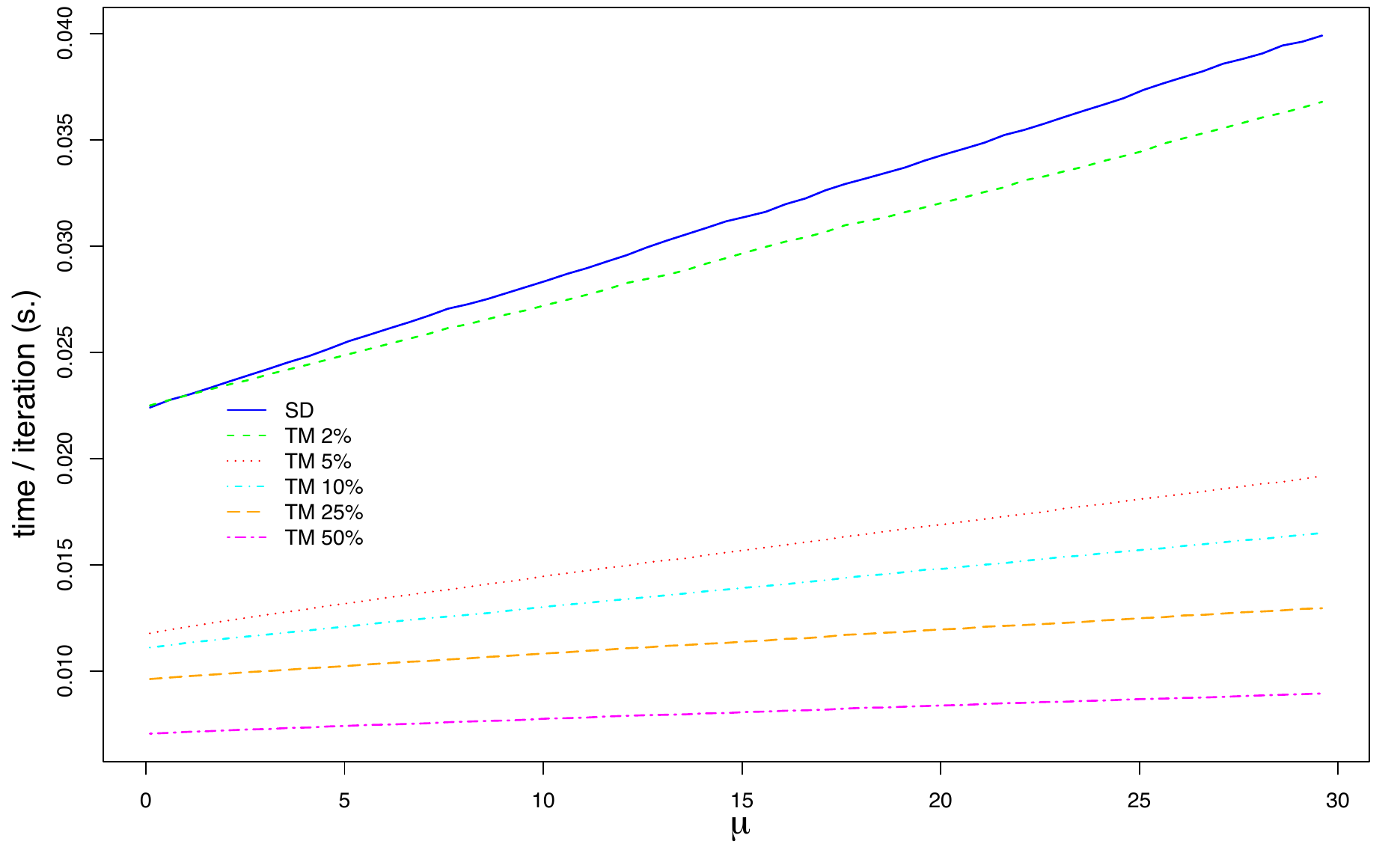}
\caption{\label{meanT} Average computation time as a function
  of the mutation rate $\mu$.  Figures are
  based on 100,000 samples.}
\end{center}
\end{figure}

\begin{figure}
\begin{center}
\includegraphics[width=\linewidth]{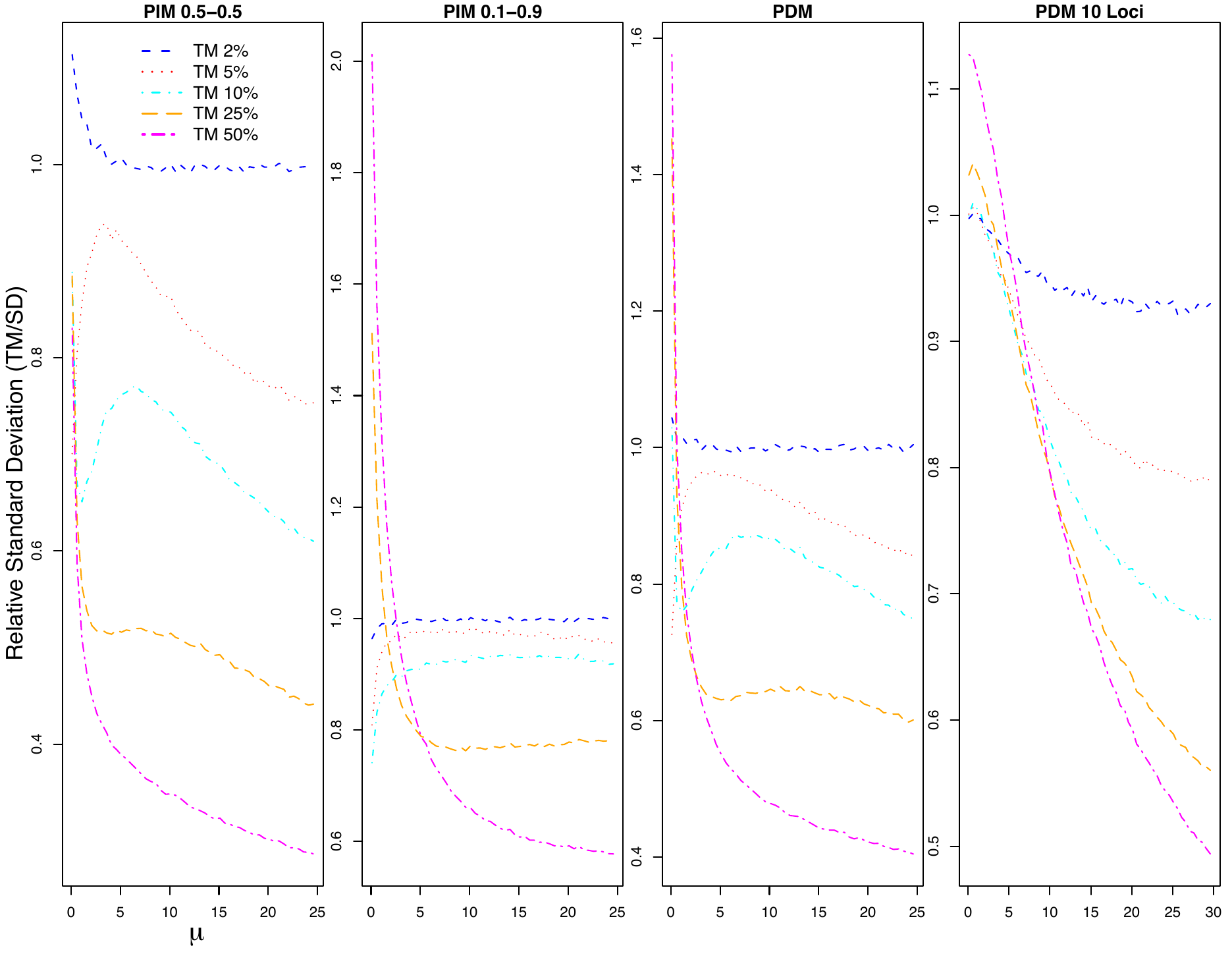}
\caption{\label{fig:relative_sd} Relative standard deviation across 100 repeats of the time machine to SD. Figures are
  based on 100,000 samples.}
\end{center}
\end{figure}

On the basis of our experiments, combining both computational efficiency and the numerical accuracy, the use of the
time machine with TM 5\% is an efficient
alternative to the SD algorithm. The C++ code is available upon request from the third author.

\section{Summary}\label{sec:summary}

In this paper we have considered a new approach for simulation
of stochastic trees and likelihood calculation of sample probabilities in population genetics models. The approach consists in stopping the backward simulations before the top of the tree is reached. We have provided theoretical results on the bias introduced in the estimation of the likelihood. Some extensions to our work are described below.

Firstly, to extend our analysis to different models. The paper has been written
to facilitate such analysis and we believe it is rather simple to deal with other stochastic tree models. Also, some further empirical investigations would help support the simulations and theoretical analyses presented here.
Our methodology would be further enhanced with GPU technology (e.g.~Lee et al.~(2010)), and this is
one area that we are currently investigating.

Secondly, to look at the consistency (in a likelihood sense) of our biased
Monte Carlo estimator. As we observed in Section \ref{sec:simos}, it appears that the Time Machine seems to recover the maximum likelihood estimator.
Therefore consistency, or potential asymptotic bias is of genuine interest.
There are very few results in the context of consistency, due to
the dependency in the data, after integrating out the tree. That is, it is
difficult to apply uniform laws of large numbers to complex dependency structures.
None-the-less, we
suggest the work of Douc et el.~(2004), Fearnhead (2003), Olsson et al.~(2008),
Olsson \& Ryd\'en (2008) as possible starting points for a proof.

Thirdly, the time machine can be used in the context of Markov chain Monte Carlo (MCMC). If one is interested in Bayesian parameter inference, then a stopping-time SMC algorithm can be used within an MCMC algorithm (particle
MCMC (Andrieu et al.~2010)). Significant time savings per iteration can be gained by using the time machine; see Jasra \& Kantas (2010) for some details.

\section*{Acknowledgements}

We thank Prof.~Arnaud Doucet for some valuable conversations related to this work.

\section*{Appendix 1: Proofs}

We give the proofs of Propositions \ref{prop:lpbound} and 
\ref{prop:biascontrol}.

\begin{proof}[Proof of Proposition \ref{prop:lpbound}]
In the case of IS, the result follows by 
adding and subtracting $l_b(y_{1:n};\theta)$
applying Minkoswki and the Marincinkiewicz-Zygmund inequality. In the case of
the SMC algorithm the proof follows from the fact that the algorithm approximates
a multi-level Feynman-Kac formula; see Chapter 12, Proposition 12.2.3 Del Moral (2004).
Note that this point is apparently over-looked in Chen et al.~(2005), and such a
result helps to verify the convergence of the algorithm. In addition, note
that the Proposition 12.2.3 of Del Moral (2004)
does not depend on the importance weights being upper-bounded by 1.
Hence, due to the boundedness
of the weights, the same proof as for IS applies, except the $\mathbb{L}_p$
bound for particle approximations of Feynman-Kac formulae is used instead
of the Marincinkiewicz-Zygmund inequality.
\end{proof}

\begin{proof}[Proof of Proposition \ref{prop:biascontrol}]
The proof is fairly simple and combines the proof of Lemma 3.9 and Theorem
4.1 of Le Gland \& Oudjane (2004) (see also Theorem 3.1 of Tadi\'c \& Doucet (2005)). The idea is to use the contraction
property of the total variation distance and Hilbert metric, as well as the
relation between the two (see Lemma 6.1 of Tadi\'c \& Doucet (2005)).

The only real complication is using the local mixing condition (\ref{eq:mixingcond})
to derive a bound on the Radon-Nikodym derivatives
$$
\frac{d(\lambda_1P_{1:k})}{d(\lambda_2P_{1:k})}, \frac{d(\lambda_2P_{1:k})}{d(\lambda_1P_{1:k})}.
$$
Consider $\lambda_1P_{1:k}$, clearly
$$
\lambda_1P_{1:k-1}(\mathbb{I}_{C_{k-1}}P_k(f)) \leq \frac{1}{\epsilon_k}\nu_k(f).
$$
In addition
$$
\lambda_1P_{1:k-1}(\mathbb{I}_{C_{k-1}^c}P_k(f)) = 0.
$$
Since 
$$
\lambda_1P_{1:k}(f) \geq \epsilon_k\nu_k(f)
$$
it follows that
$$
\frac{d(\lambda_1P_{1:k})}{d(\lambda_2P_{1:k})} \leq \frac{1}{\epsilon_k^2}.
$$
The proof can then be concluded by following the arguments of Le Gland \& Oudjane (2004), Lemma 3.9 and Theorem 4.1.
\end{proof}

\section*{Appendix 2: Algorithm Description}
Let $x_t=(x_{1,t},\dots,x_{d,t})$, $t\in \mathbb{Z}^+$, the population
size within each of the $d$ states at time $t$. The algorithm will
simulate backward in time genealogical trees for an initial population,
$x_1$ the ($d-$dimensional) counts associated to the observed data,
until there are $N_{TM}$ sequences left in the population. The case
where $N_{TM}$=$1$ corresponds to ordinary coalescent and $N_{TM}>1$
to the time machine. Most of the notations can be found in Sections
\ref{sec:motivation} and \ref{sec:id_interest}.

\subsection*{Iterative algorithm}
For any generation $t$, there are $|x_t|_1$ sequences left in the
population, the following steps will be iterated until $|x_t|_1$=$N_{TM}$:
\begin{enumerate}
\item Sampling the type of the offspring sequence ($i$) with
  probability
$$
\frac{x_{i,t}}{|x_t|_{1}}.
$$
\item Getting the type of the ancestor sequence ($j$).\\
A sequence of a given type $i$ can have arisen from an ancestor
sequence of type $j$ through:
\begin{enumerate}
\item a coalescent event with a probability proportional to 
$$
|x_{i,t}| - 1.
$$
\item a $j$ to $i$ mutation event (inclusive of self mutations, from type $i$ to type $i$), with probability proportional to:
\begin{equation}
\mu \kappa_{ij} p_{ji},
\label{P_mut}
\end{equation}
where
\begin{equation*}
  \kappa_{ij}=\left\{
    \begin{array}{cl}
      {\displaystyle{\frac{x_{j,t}+\mu \psi_{j}}{|x_t|_1 - 1 +
            \mu}}}&\mbox{if $j\neq i$}\\
\\
{\displaystyle{\frac{x_{j,t}-1+\mu \psi_{j}}{|x_t|_1 - 1 + \mu}}}&\mbox{if $i$=$j$}
   \end{array}\right.,
\end{equation*}
\end{enumerate}
\item Updating the population sizes within each type.\\
\begin{equation*}
  x_{t+1}=\left\{
    \begin{array}{cl}
 x_t - e_i + e_j      &\mbox{if a mutation has occurred }\\
\\
 x_t - e_i&\mbox{if a coalescent event was simulated}
   \end{array}\right.,
\end{equation*}
\item Calculate the contribution to the likelihood of the simulated
  event (suppressing the subscript $\theta$)\\
\begin{equation}
 w_t=\left\{
    \begin{array}{cl}
{\displaystyle{\frac{K_1}{K_2} \frac{\kappa_{ii}}{\kappa_{ij}} \frac{x_{j,t+1}}{|x_{t}|_1}}}      &\mbox{if a mutation has occurred }\\
\\
{\displaystyle{\frac{K_1}{K_2} \frac{1}{\kappa_{ii}} \frac{x_{i,t+1}(|x_{t+1}|_1-1)}{x_{i,t}(x_{i,t}-1)}}}&\mbox{if a coalescent event was simulated}
   \end{array}\right.,
   \label{W_t_mut_bis}
\end{equation}
where
$$
K_1=|x_t|_1  (|x_t|_1-1+\mu),
$$
and
$$
K_2=|x_{t+1}|_1 (|x_{t+1}|_1-1+\mu).
$$
\item Updating the log likelihood\\
\begin{equation}
 W_{t}=\left\{
    \begin{array}{cl}
\log(w_t)      &\mbox{if $t$=0}\\
\\
W_{t-1}+\log(w_t)&\mbox{it $t\geq 1$}
   \end{array}\right.
\label{update_L}
\end{equation}
\item Assessing the stopping criterion.\\
  When the time machine is used ($i.e.$ $N_{TM}>1$), steps 1 to 5 are
  repeated until $|x_{t+1}|_{1}>N_{TM}$. Otherwise, when the full tree
  is simulated ($N_{TM}$=$1$), steps 1 to 5 are repeated until there are 2
  sequences left in the population. Then, mutations are simulated
  until both remaining sequences are of the same type, based on the
  following three steps:
\begin{enumerate}
\item Choose one of the two sequence, of type $i$, with probability 0.5.
\item Simulate the mutation event from an ancestor of type $j$ (to
  type $i$) according to the probability defined in equation
  (\ref{P_mut}), and setting the coalescent probability to 0.
\item Calculate the corresponding weight for the sampled $j$ to $i$
  simulated transition. At this final stage there are only two
  individuals in the population ($|x_t|_1=|x_{t+1}|_1=2$), hence
  $K_1=K_2$, and $x_{i,t}=1$, then:
$$
 \frac{\kappa_{ii}}{\kappa_{ij}} =\frac{\mu \psi_i}{x_{j,t}+\mu \psi_j}.
$$
When the final generation is reached $x_{j,t+1}=2$, and $x_{j,t+1}=1$
for any other iterations in that step. The weight for each generation,
derived from \eqref{W_t_mut_bis} is then defined as:
\begin{equation*}
w_t=\left\{
\begin{array}{cl}
{\displaystyle{\frac{\mu \psi_i}{x_{j,t}+\mu \psi_j}}}&\mbox{at the last generation}\\
{\displaystyle{\frac{\mu \psi_i}{2(x_{j,t}+\mu \psi_j)}}}&\mbox{otherwise}
\end{array}\right.
\end{equation*}
\item Update the likelihood according to equation (\ref{update_L})
\end{enumerate}
\end{enumerate}

\subsection*{Estimating the bias}
This step is specific to the time machine (i.e.~if $N_{TM}\geq
1$). Recalling that $\rho$ is the generation at which the iterative algorithm
was stopped, the bias induced by stopping the simulation before
reaching the MRCA is estimated as:
\begin{equation*}
  \log(b)=\log\left(\frac{(|x_{\rho}|_1)!\Gamma(\mu)}{\Gamma(\mu+|x_{\rho}|_1)}\right)+\sum_{i=1}^d\log\left(
  \frac{\Gamma(x_{i,\rho}+\mu \psi_i)}{(x_{i,\rho})!\Gamma(\mu \psi_i)}\right),
\end{equation*}
where $\Gamma$ denotes the gamma function.
The likelihood of the tree is then updated
$$
W_{\rho+1}=W_{\rho}+\log(b).
$$
\subsection*{Estimation of the likelihood}
The above algorithm is independently repeated $N$ times, the estimate of the log-likelihood is
$$
\log\bigg\{\frac{1}{N}\sum_{i=1}^Ne^{W^{(i)}-\widetilde{W}}\bigg\} + 
\widetilde{W},
$$
where $W^{(i)}$ is the value of the final weight for sample $i$ and
$\widetilde{W}=\max_{1\leq i \leq N} W^{(i)}$. 

\section*{Appendix 3: Infinite Sites Model}

We now consider our results in the context of the infinite
sites model. We concentrate upon
likelihoods associated to rooted genealogical trees; see Ethier \& Griffiths (1987) or Griffiths \& Tavar\'e (1995) for more details.

\subsection*{The Model}
The model is based upon the simulation of distinct DNA sequences,
and the multiplicity of the sequences. In more details, the simulation
begins with a single DNA sequence $x_1=(0)$, and counts $n_1=2$. The process
can then undergo a mutation (rate $\mu$) or a split. If a mutation occurs (to the first sequence say)
we have the new state $x_1=(1,0)$, $x_2=(0)$ and $n=(1,1)$, otherwise the
new state is $x_1=(0)$ and $n_1=3$. 

The key point is that new mutations introduce
a new site (that is a new integer number (which is larger than all others
currently present) to the start of a selected sequence) and hence DNA sequence, whilst splits only increase the number of an existing
sequence. The state-space consists of the $d-$distinct sequences (vectors of potentially
different length sequences
$x_1,\dots,x_d$) and the respective counts $(n_1,\dots,n_d)$ of the sequences
that have been simulated. That is, in the previous notation
$$
z_i = ((x_{1:d})_i,(n_{1:d})_i).
$$
The simulation stops, as before, when $|(n_{1:d})_k|_1=n+1$.
In general, transitions are governed by the following Markov kernel. A mutation
(rate $\mu$), at time $j-1$, of the $l^{th}$
sequence occurs with probability
$$
\frac{1}{|(n_{1:d})_{j-1}|_1}\frac{\mu}{|(n_{1:d})_{j-1}|_1-1+\mu}
$$
and a split of the $l^{th}$ sequence occurs with probability
$$
\frac{|(n_{1:d})_{j-1}|_{1}-1}{|z_{j-1}|_1}\frac{|(n_{1:d})_{j-1}|_1-1}{|(n_{1:d})_{j-1}|_1-1 + \mu}
$$
see Ethier \& Griffiths (1987) and Griffiths (1989) for details on the transition dynamics.

In this scenario, the state-space is more complicated. 
Let
\begin{eqnarray*}
E_{d,r_{1}:r_d} &  = & \{(r_1,\dots,r_d)\in(\mathbb{Z}^+\cup\{0\})^d, (x_{1:r_1}^{1},\dots,
x_{1:r_d}^{d})\in (\mathbb{Z})^{r_1}\times\cdots\times ((\mathbb{Z})^{r_d}):
\\ & &  x_{1:r_1}^1\neq\cdots\neq x_{1:r_d}^d, i\in\mathbb{T}_d, x_{1}^{i}>\cdots>
x_{r_i}^{i}=0\}\;
\end{eqnarray*}
here $r_i$ are the lengths of the distinct sequences, and the ordering constraint
notes that the discovery of a new site is added to the beginning of the segment
vector.
In addition, let
\begin{eqnarray*}
F_{n,d} & = & \{n^{1:d}\in(\mathbb{Z}^+)^d\cap 2\leq |n^{1:d}|_1\leq n+1\}\\
E_n & = & \bigcup_{d_n\in\mathbb{T}_{n+1}}\{d_n\}\times E_{d_n,r_{1}:r_{d_n}}\times F_{n,d_n}
\end{eqnarray*}
then 
$$
F = \bigcup_{k\in\mathcal{K}_n}\{k\}\times E_n^{k}.
$$
There are three trans-dimensional aspects to the state-space; the time to simulate
$n+1$ sequences; the number of distinct sequences and the respective lengths
of the distinct sequences (which is determined in part by the first two aspects).

\subsection*{The Bias}

For the infinitely-many-sites model, we will use the idea of the 
first time the
number of segregating sites is $m$ (or mutations here) to stop the simulations backward
in time. In a similar manner to Section \ref{sec:stopcoal}, it can be established
that we want the approximating function $h_{\theta}(\cdot)$ to be the marginal
of the process at the last time we have $m$ segregating sites.

In the context of the infinitely-many-sites model, the bias is controlled
by our ability to approximate this marginal (see Remark 2 in Section \ref{sec:contrbias}). This is because the Markov transitions
can only change the multiplicity of counts, or increase the number of distinct
sequences; we are unable to change the beginning of sequences. As a result, it is not possible to establish conditions such as (A\ref{hyp:p_nassump}).

\subsection*{Approximating the Marginal}

We propose the following approximation of the marginal, based upon the theoretical
properties of such models (Ethier \& Griffiths, 1987;Griffiths, 1989) and the relation to the infinitely-many-alleles model (e.g.~Griffiths (1979) and the references
there-in). Let us consider the marginal distribution, call it $\xi_{\theta}$.
We extend the state-space to include uncertainty on $d$, the number of distinct
types, $c=|n_{1:d}|_1$ and $s$ the number of segregating sites, and adopt
the decomposition
$$
\xi_{\theta}(z_{1:d},d,s,c) = \xi_{\theta}(x_{1:d}|d,s)
\xi_{\theta}(n_{1:d}|d,c)\xi_{\theta}(d|c)\xi_{\theta}(s|c)\xi_{\theta}(c).
$$
Now, under certain conditions, there are results about the exact
densities $\xi_{\theta}(n_{1:d}|d,c)$ (Ewens, 1972) and $\xi_{\theta}(d|c)$ (Watterson, 1975). In the case $\xi_{\theta}(s|n)$, as noted by Griffiths (1979),
for large populations (such that diffusion results can apply) the infinitely-many-sites and infinitely-many-allele frequencies are not too different. Therefore,
we propose to use the probability (as in Ewens (1972))
$$
\xi_{\theta}(d|c) = \frac{\mu^d\Gamma(\mu)}{\Gamma(\mu+c)}|S_{c}^{(d)}|
$$
with $S_{c}^{(d)}$ are Stirling numbers of the first kind.

For the quantities $\xi_{\theta}(x_{1:d}|d,s)$ and $\xi_{\theta}(c)$,
we use approximations. For the former, a uniform distribution
is adopted
$$
\xi_{\theta}(x_{1:d}|d,s) = \bigg(\sum_{m_1,\dots,m_{s}\in C_{d,s}}
\bigg[\prod_{i=1}^s\binom{d}{m_i}\bigg]\bigg)^{-1}
$$
where
\begin{eqnarray*}
C_{d,s} & := & \{m_{1:s} : m_i\in\{0,\dots,d\},m_1+\dots+m_{s}\in \bar{C}_{d,s}\}\\
\bar{C}_{d,s} & := & \{d-1,\dots,\sum_{j=0}^{d-1}(s-j)\}.
\end{eqnarray*}
That is, it is a simple task in combinatorics to show that if there are 
$s$ mutations with $m_{1:s}$ repetitions of mutations $1$ to $s$, 
(subject to the constraint that each mutation can only occur at most once in each sequence and that order of allocating a mutation does not matter) then there are 
$$
\prod_{i=1}^s\binom{d}{m_i}
$$
possible sequences; summing over all the possible multiples yields the desired
cardinality of the state-space. $\xi_{\theta}(c)$ is not known (except as
the marginal of a recursion (as in Ethier \& Griffiths (1987))) and is assigned
$\mathcal{P}ois(nj/\theta)$ (Poisson) distribution (at time $j$). 

In practice, it may not be possible to evaluate some of these quantities
and a further Monte Carlo simulation/numerical approximation (for the integral
over $s$ and the normalizing constant of $\xi_{\theta}(x_{1:d}|d,l)$) will be required. That is to say, we set
$$
h_{\theta}(z_{1:d}) = \widehat{\xi}_{\theta}(z_{1:d}).
$$
The approximation will be different for every simulated sample.

\end{document}